\documentclass[11pt,a4paper]{article}
\usepackage[affil-it]{authblk}
\usepackage{fullpage}
\usepackage{graphicx}
\usepackage{amssymb,amsmath,amsthm}
\usepackage{amsrefs}
\usepackage{epstopdf}
\usepackage[colorlinks=true, citecolor=blue]{hyperref}
\usepackage{color}
\usepackage{array}
\usepackage[mathscr]{euscript}
\usepackage{caption, subcaption}

\newtheorem{lem}{Lemma}

\newtheorem{thm}{Theorem}
\newtheorem*{thm*}{Theorem}
\newtheorem{prop}{Proposition}
\newtheorem{conj}{Conjecture}

\title{A solvable non-directed model of polymer adsorption}
\author{Nicholas R Beaton}
\affil{Department of Mathematics and Statistics, University of Saskatchewan, Saskatoon, Canada}
\author{Gerasim K. Iliev}
\affil{Department of Chemistry, University of Toronto, Toronto, Canada}
\date{}

\def\e{e}
\def\alphac{\alpha_{\rm c}}
\def\ac{a_{\rm c}}

\def\aSl{a_{c}^{{\ell}}}
\def\aSt{a_{c}^{{t}}}

\def\tSl{z_{}^{{\ell}}}
\def\tSlone{z_{1}^{{\ell}}}
\def\tSltwo{z_{2}^{{\ell}}}

\def\tSt{z_{}^{{ t}}}
\def\tStone{z_{1}^{{ t}}}
\def\tSttwo{z_{2}^{{ t}}}

\makeatletter
\def\blfootnote{\xdef\@thefnmark{}\@footnotetext}
\makeatother

\begin{document}
\maketitle

\begin{abstract}
Prudent walks are self-avoiding walks which cannot step towards an already occupied vertex. 
We introduce a new model of \emph{adsorbing} prudent walks on the square lattice, 
which start on an impenetrable surface and accrue a fugacity $a$ with each step along the surface. 
These are different to other exactly solved models of polymer adsorption, like Dyck paths, 
Motzkin paths and partially-directed walks, in that they are not trivially directed - 
they are able to step in all lattice directions. We calculate the generating functions, free energies and 
surface densities for this model and observe a first-order adsorption transition at the critical 
value of the surface interaction.\blfootnote{Email: {\tt n.beaton@usask.ca}, {\tt giliev@chem.utoronto.ca}}
\end{abstract}

\section{Introduction}
Lattice walk models have long been considered as a canonical model for long-chain 
polymers in solution~\cite{Orr1947Statistical}. {One popular lattice walk model in the polymer community 
is that of self-avoiding walks (SAWs).} The standard model associates a weight (or fugacity) $z$ 
with each step, or monomer. If $c_n$ is the number of $n$-step walks (equivalent up to translation), 
the length generating function is
\[F(z) = \sum_{n\geq0}c_n z^n.\]
It has long been known that the limit
\[\log\mu = \lim_{n\to\infty}\frac1n \log c_n\]
exists, where the \emph{growth constant} $\mu$ depends on the lattice being considered.  
For most lattices in 2 or higher dimensions, only numerical estimates for $\mu$ are 
known~\cite{Jensen1999Selfavoiding, Jensen2004Selfavoiding}. 
A notable exception is the 2-dimensional honeycomb lattice, for which it has recently been 
proven by Smirnov and Duminil-Copin~\cite{DuminilCopin2012Connective} that $\mu=\sqrt{2+\sqrt{2}}$. 

To model the behaviour of a polymer interacting with an 
impenetrable surface, we consider walks restricted to a half-space. Although this restriction
limits the number of walk configurations, it has no effect on the exponential growth of the
number of walks and hence the growth constant~\cite{Whittington1975Selfavoiding}. 
To model the monomer-surface interactions of an adsorbing polymer, we associate a 
fugacity $\e^{\alpha}$ with every edge of the walk which lies in the boundary
of the half-space.\footnote{{One can instead associate the fugacity with \emph{vertices} in the surface, 
to obtain a slightly different model.}}  If we define $c_n^+(\nu)$ as the number of
$n$-step half-space walks beginning on the surface with $\nu$ edges along the 
surface, then the partition function is given by
\[Z_n(\alpha) = \sum_{\nu=0}^n c_n^+(\nu)\e^{\nu\alpha},\]
where $\alpha= -\epsilon/k_B T$.  Here $\epsilon$ is the energy associated with a surface 
visit, $T$ is the absolute temperature and $k_B$ is Boltzmann's constant.
For sufficiently large values of $\alpha$, configurations with many visits to the surface dominate
the partition function and the walk is said to be in an \emph{adsorbed} state;  
otherwise, the loss in configurational entropy dominates, and the walk is repelled by 
the surface and said to be in a \emph{desorbed} state. 

The free energy of the system is 
\[f(\alpha) = \lim_{n\to\infty}\frac1n\log Z_n(\alpha),\]
which has been shown~\cite{Hammersley1982Selfavoiding} to exist for all $\alpha < \infty$. When $\alpha<0$ the 
free energy is independent of $\alpha$ and given by $\log\mu$~\cite{Whittington1975Selfavoiding}. 
For $\alpha\geq0$, it has been shown that in two dimensions
\[f(\alpha) \geq \max\{\log\mu,\alpha\}.\]
This implies the existence of a critical value $\alpha_{\rm c}$ with 
$0 < \alpha_{\rm c} < \log\mu$ where the free energy 
$f(\alpha)$ is non-analytic. We identify this point of non-analyticity as the location of the 
\emph{adsorption phase transition:} for $\alpha<\alpha_{\rm c}$, the polymer is desorbed, 
while for $\alpha>\alpha_{\rm c}$ the polymer is adsorbed.

Closely related to the free energy is the mean density of surface visits; 
in the limit of infinitely long polymers this is given by
\begin{equation}\label{eqn:density}
\rho(\alpha) = \frac{\partial f(\alpha)}{\partial\alpha}.
\end{equation}
This variable acts as an order parameter for the system, and signals the 
onset of a phase transition.
If $\alpha<\alpha_{\rm c}$ then $\rho(\alpha)=0$, while $\rho(\alpha)>0$ for 
$\alpha>\alpha_{\rm c}$. If the adsorption transition is \emph{second-order} then 
$\rho$ is continuous for all $\alpha$, while a \emph{first-order} transition is
manifested by a jump discontinuity in $\rho(\alpha)$ at $\alpha_{\rm c}$.

Instead of working directly with the partition functions, $Z_n$, 
we will consider the generating function $Z(z,\alpha) = \sum_n Z_n(\alpha) z^n$. 
A central tenet of analytic combinatorics~\cite{Flajolet2009Analytic} is that the 
\emph{dominant singularities} (the points of non-analyticity closest to the origin) 
of a generating function determine the asymptotic behaviour of its coefficients. 
Specifically, if for a given $\alpha$ the dominant singularity of $Z(z,\alpha)$ 
occurs at $z = z_{\rm c}(\alpha)$, then 
\begin{equation}\label{eqn:fe_domsing}
f(\alpha) = \lim_{n\to\infty} \frac1n \log Z_n(\alpha) = -\log z_{\rm c}(\alpha).
\end{equation}
For many models (including the one considered in this paper) it is much easier to determine 
the generating function than to compute the individual partition functions $Z_n$ for 
each $n$. But since~\eqref{eqn:fe_domsing} enables us to determine the free energy 
$f(\alpha)$ of a model directly from its generating function, we need never consider 
the $Z_n$ anyway.

For ease of notation we will frequently use the surface fugacity $a={\rm e}^{\alpha}$, 
and consider the free energy, etc. as functions of $a$. 
The above arguments are still valid, with the critical fugacity occuring at 
$\ac = e^{\alphac}$.

\subsection{Exactly solved models}\label{ssec:exactly_solved}
At the present time, a closed form expression for the generating or partition function of SAWs is still unattainable. 
Numerical experiments, using series analysis~\cite{Clisby2012New, Guttmann2014Pulling} or Monte Carlo simulations~\cite{Clisby_2013, vanRensburg2004Multiple}, can be 
used to estimate non-rigorous values of $\mu$ and $\ac$, for example.

Exactly solvable models can be obtained by placing a restriction on the walk models used to
study the polymer system. Possibly the simplest restriction to use is \emph{directedness -} 
forbidding steps in certain lattice directions. Below, we summarize the three most 
commonly-used directed models.
\begin{itemize}
\item {\bf Dyck\slash ballot paths.} \cite{vanRensburg1999Adsorbing} These are generated on the ($45^{\circ}$ rotated) 
square lattice by forbidding NW and SW steps. Since there are no steps parallel to the surface, only 
vertex-weighted models are possible. In the desorbed phase, these walks have growth constant $\mu=2$. 
The adsorption phase transition occurs at $\ac=2$.
\item {\bf Motzkin paths.} \cite{vanRensburg2010Adsorbing} These are generated on the triangular lattice by forbidding 
NW, W and SW steps. The growth constant is $\mu=3$, and the phase transition occurs at $\ac = 2$ for 
edge weights and $\ac=3/2$ for vertex weights.
\item {\bf Partially-directed walks.} \cite{Privman1989Directed} These are generated on the square lattice by 
forbidding W steps. The growth constant here is $\mu=1+\sqrt{2}$, and the phase transition occurs at 
$\ac = (2+\sqrt{2})/2$ for edge weights and $\ac = (1+\sqrt{2})(\sqrt{5}-1)/2$ for vertex weights.
\end{itemize}
\begin{figure}
\centering
\begin{subfigure}[b]{0.3\textwidth}
\begin{picture}(0,0)
\put(0,1.4){\includegraphics[width=\textwidth]{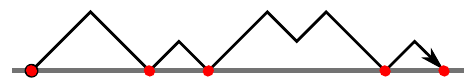}}
\end{picture}
\caption{}
\end{subfigure}
\begin{subfigure}[b]{0.3\textwidth}
\includegraphics[width=\textwidth]{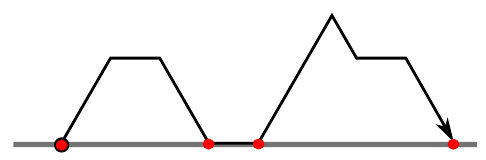}
\caption{}
\end{subfigure}
\begin{subfigure}[b]{0.3\textwidth}
\begin{picture}(0,0)
\put(0,0.5){\includegraphics[width=\textwidth]{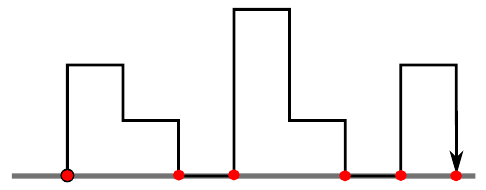}}
\end{picture}
\caption{}
\end{subfigure}
\caption{Examples of directed polymer adsorption models, with vertex weights indicated: 
(a) a Dyck path, (b) a Motzkin path and (c) a partially-directed walk.}
\label{fig:directed_examples}
\end{figure}

All of these directed models undergo second-order adsorption transitions at their critical values of
the surface fugacity. These models are relatively easy to solve, and can easily be adapted to model other 
polymeric objects and phenomena, such as copolymers and polymer 
collapse~\cite{Madras2003Localization, Brak1992Collapse, Owczarek2007Exact}. 
The directedness restrictions are, however, very strong, and result in walks which have considerably 
less freedom than general SAWs.

Here we introduce a model of polymer adsorption on the square lattices which does not 
require a directedness restriction -- the walks are able to step in all lattice directions. 
Instead of being directed, the walks must be \emph{prudent} -- they can never take a step towards a previously-visited 
lattice point.
The model of partially-directed walks can be viewed as a subset of this prudent walk model. 

Using the prudence restriction we recursively generate the walks by using a set of generating functions 
that keep track of certain geometric measurements. These recurrences can be written in the form of 
functional equations, which can be solved via the \emph{kernel method} and some of its generalisations. 

We analyze the generating functions to determine the free energy $f(\alpha)$ and the surface density $\rho(\alpha)$
for prudent walks.  These calculations are carried out for walks where either one or both endpoints of the walk 
lie in the surface and we refer to these two subsets as \emph{tails} and \emph{loops}, respectively.
We find that both prudent tails and loops undergo a \emph{first-order} adsorption transition at
$\aSt=2$ and $\aSl\approx1.82476$, respectively.

\subsection{Prudent walks}
Prudent walks on the square lattice were introduced by Pr\'ea \cite{Prea1997Exterior} and later studied by Duchi 
\cite{Duchi2005Some} and Bousquet-M\'elou \cite{BousquetMelou2010Families}. These walks have the useful property 
that their endpoint always lies on the boundary of their \emph{bounding box}  -- the smallest 
lattice rectangle which contains the walk. Using this property,  Duchi defined four classes of 
increasing generality:
\begin{itemize}
\item \emph{1-sided} prudent walks must always end on the E side of their bounding box,
\item \emph{2-sided} prudent walks must always end on the N or E sides,
\item \emph{3-sided} prudent walks must always end on the N, E or W sides (and may not step between 
the SE and SW corners of the box when the box has width one), and
\item \emph{4-sided} (or \emph{general}) prudent walks may end anywhere on their box.
\end{itemize}

\begin{figure}
\centering
\begin{subfigure}[b]{0.3\textwidth}
\includegraphics[width=\textwidth]{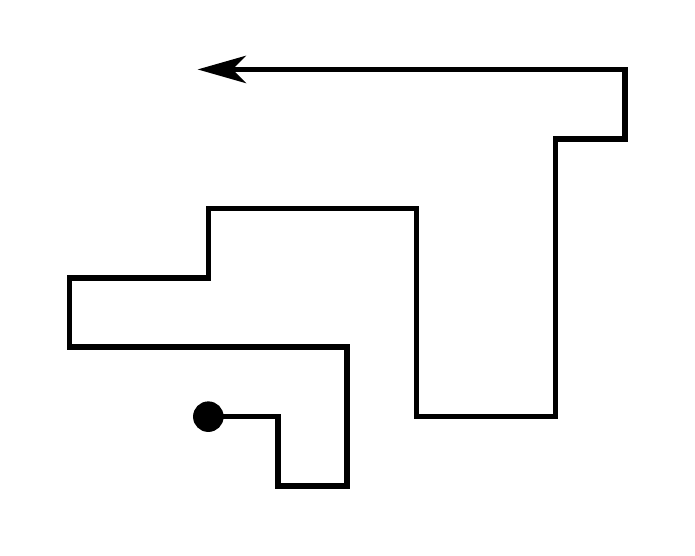}
\caption{}
\end{subfigure}
\begin{subfigure}[b]{0.3\textwidth}
\includegraphics[width=\textwidth]{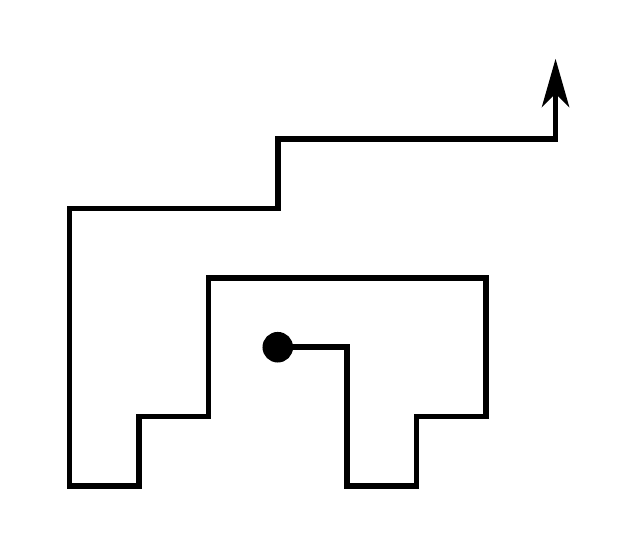}
\caption{}
\end{subfigure}
\begin{subfigure}[b]{0.3\textwidth}
\includegraphics[width=\textwidth]{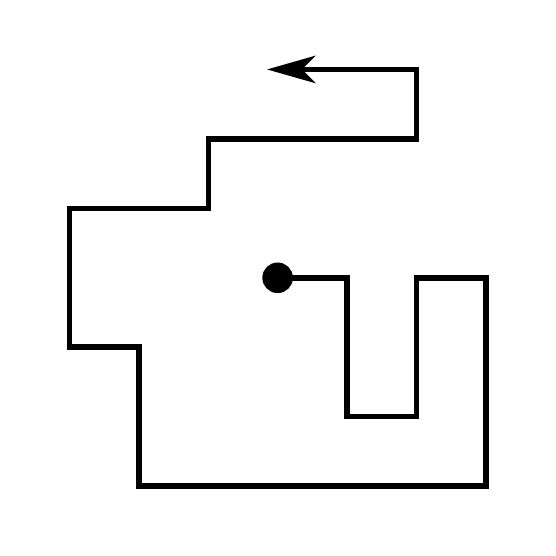}
\caption{}
\end{subfigure}
\caption{Prudent walks on the square lattice: (a) 2-sided, (b) 3-sided and (c) 4-sided (unrestricted).}
\label{fig:square_pruwalks}
\end{figure}

The class of 1-sided prudent walks are equivalent to partially-directed walks, a model of polymer 
adsorption that has been studied for some time (for an early example see~\cite{Privman1989Directed}). 
Duchi and Bousquet-M\'elou solved the generating functions of 2- and 3-sided prudent walks 
(without a boundary) respectively. The generating function of 2-sided walks is algebraic, while that 
of 3-sided walks is non-D-finite. Bousquet-M\'{e}lou also wrote down a functional equation for the 
generating function of general (4-sided) prudent walks, but no solution for the generating function has
yet been found.

\section{Generating functions}

We are able to solve the functional equation for 2-sided prudent walks above an impenetrable surface, however
the addition of the surface and the corresponding interaction fugacity $a$ makes the generating function significantly 
more complex than that of the original (surface-free) model.  Determining the critical critical behaviour 
of the adsorption model requires a great deal more effort. 

We have already noted that 1-sided prudent walks are equivalent to partially-directed walks.
For the 3-sided model, we are able to write down functional equations, but the number of catalytic 
variables is too great to obtain a closed form solution. 
There are no 4-sided walks which are not also 3-sided walks. To see this, note that the S side of 
the bounding box necessarily lies on the impenetrable surface. So any step taking the walk from the E or W sides onto the S side must be in the direction of the origin, and thus cannot possibly be prudent.

For the 2-sided purdent walk model, we begin by defining the generating functions 
\begin{align*}
R(u,v) &:= R(z;u,v;a) = \sum_{n,i,j,\nu}R_{n,i,j,y}z^n u^i v^j a^\nu\\
T(u,v) &:= T(z;u,v;a) = \sum_{n,i,j,\nu}T_{n,i,j,y}z^n u^i v^j a^\nu
\end{align*}
where $R_{n,i,j,\nu}$ (resp. $T_{n,i,j,\nu}$) is the number of $n$-step 2-sided prudent walks which 
start on a horizontal impenetrable surface and end on the right (resp. top) side of their bounding box, 
with a distance $i$ from the endpoint to the top (resp. right) of the box, a distance $j$ from the endpoint 
to the surface, and $\nu$ steps along the surface. We write $R(u,v)$ and $T(u,v)$ for shorthand as $u$ and $v$ 
are the \emph{catalytic} variables -- the ones used in the iterative construction of the walks.
The introduction of the surface breaks the symmetry between walks ending on the top and right of the box, 
forcing us to use two generating functions.

\begin{figure}
\centering
\begin{subfigure}[b]{0.4\textwidth}
\centering
  \begin{picture}(160,80)
    \put(0,0){\includegraphics[scale=1.2]{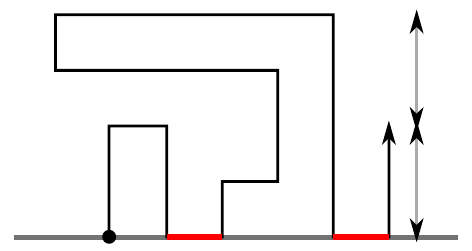}}
    \put(152,25){$j$}
    \put(152,60){$i$}
 \end{picture}
 \caption{}
\end{subfigure}
\hspace{0.1\textwidth}
\begin{subfigure}[b]{0.4\textwidth}
\centering
  \begin{picture}(140,80)
    \put(0,0){\includegraphics[scale=1.2]{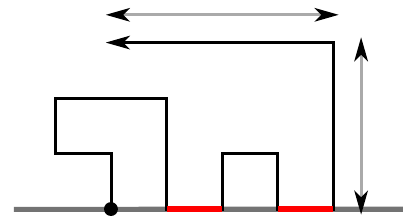}}
    \put(133,30){$j$}
    \put(75,80){$i$}
  \end{picture}
  \caption{}
  \end{subfigure}

\caption{2-sided prudent walks above an impenetrable surface, (a) ending on the right of the box 
and (b) ending on the top, with the distances $i$ and $j$ indicated.}
\label{fig:sq_prudent_surface_catalytic}
\end{figure}
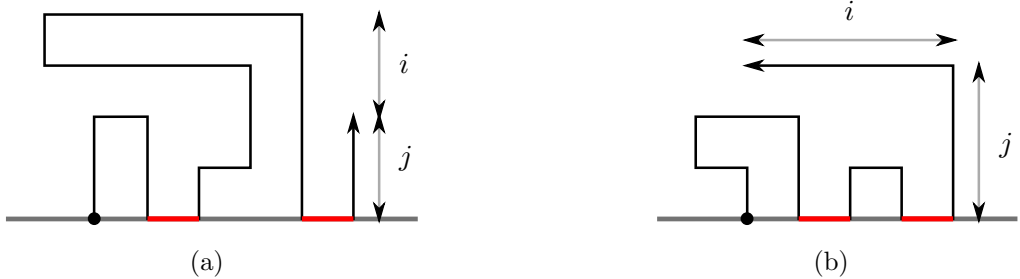

\begin{lem}\label{lem_square_original_func_eqns}
The generating functions $R(u,v)$ and $T(u,v)$ satisfy the functional relations
\begin{equation} \label{eqn:T_eqn}
L(u,v)T(u,v) = \frac{1}{1-zua} - \frac{z^2v}{u-z}T(z,v)+zR(zv,v)-z(1-a)R(zv,0)
\end{equation}
\begin{equation}\label{eqn:R_eqn}
\begin{split}
M(u,v)R(u,v) = 1+zvT(z,v)-\frac{z^2v}{u-zv}R(zv,v)-\frac{z^2u}{v-zu}R(u,zu)-\frac{zu(1-a)}{u-zv}R(u,0)\\+\frac{z^2v(1-a)}{u-zv}R(zv,0)
\end{split}\end{equation}
where
\begin{align*}
L(u,v) &:= L(z;u,v) = 1-\frac{zuv(1-z^2)}{(u-z)(1-zu)}\\
M(u,v) &:= M(z;u,v) = 1-\frac{zuv(1-z^2)}{(v-zu)(u-zv)}
\end{align*}
\end{lem}

\begin{proof} We iteratively construct walks by considering their \emph{last inflating step} -- the last step which moved either the top or right side of the bounding box. We begin with the construction for walks ending at the top of their box. There are three possibilities:
\begin{itemize}
\item {\bf Case 1.} Neither the top nor right sides of the box have ever moved. Then the walk is empty or consists only of W steps. The generating function for such walks is
\[1+zua+z^2u^2a^2+\ldots = \frac{1}{1-zua}.\]
\item {\bf Case 2.} The right side was the last to move. Then immediately prior to this, the walk ended on the right side of its box, and it then took an E step. After this it must have taken N steps to reach the top of the box, but no further. There are, however, two sub-cases to consider, depending on whether that E step was along the surface. If it was then the walk gains an $a$ weight; otherwise it does not.
\subitem {\bf Case 2a.} The inflating E step was along the surface. The generating function for such walks is
\[\sum_{n,i,\nu}R_{n,i,0,\nu} z^{n+i+1} v^{i}a^{\nu+1} = za R(zv,0).\]\
\subitem {\bf Case 2b.} The inflating E step was not along the surface. The generating function for such walks is then
\[\sum_{\substack{n,i,j,\nu \\ j>0}} R_{n,i,j,\nu} z^{n+i+1} v^{i+j}a^\nu = z R(zv,v) - z R(zv,0).\]
\item {\bf Case 3.} The top side was the last to move. Then immediately prior to this, the walk ended on the top side of its box, and it then took a N step. After this, it may have taken E steps as far as the right side of the box, or any number of W steps. The generating function for such walks is then
\begin{align*}
&\sum_{n,i,j,\nu}T_{n,i,j,\nu}v^{j+1}a^\nu \left(\sum_{k=0}^i z^{k+1}u^{i-k} + \sum_{l=1}^\infty z^{l+1}u^{i+l}\right)\\
\intertext{(where the sum over $k$ counts the walks which go E after the N step, and the sum over $l$ counts those which go W after the N step),}
&= \sum_{n,i,j,\nu}T_{n,i,j,\nu}v^{j+1}a^\nu \left(\frac{z(u^{i+1}-z^{i+1})}{u-z} + \frac{z^2u^{i+1}}{1-zu}\right)\\
&= \frac{zv}{u-z}\left(uT(u,v) - zT(z,v)\right) + \frac{z^2uv}{1-zu}T(u,v).
\end{align*}
\end{itemize}
Every walk counted by $T(u,v)$ falls into exactly one of Cases 1-3. We thus have
\begin{multline*}
T(u,v) = \frac{1}{1-zua} + zaR(zv,0) +zR(zv,v)-zR(zv,0) \\+\frac{zv}{u-z}\left(uT(u,v) - zT(z,v)\right) + \frac{z^2uv}{1-zu}T(u,v).
\end{multline*}
Rearranging this gives~\eqref{eqn:T_eqn}.

A similar consideration can be made for walks ending on the right side of their box. The same three cases apply.
\begin{itemize}
\item {\bf Case 1.} Neither the top nor right sides of the box have moved. The only such walk is the empty walk, with generating function 1.
\item {\bf Case 2.} The right side of the box was the last to move. Then the last inflating step was E. After this, the walk may take S steps as far as the surface or N steps as far as the top of the box.
\subitem {\bf Case 2a.} The last inflating step was along the surface, and the walk gained an $a$ weight. The generating function for such walks is
\[\frac{za}{u-zv}\left(uR(u,0) - zvR(zv,0)\right).\]
\subitem {\bf Case 2b.} The last inflating step was not along the surface. The generating function for such walks is
\begin{multline*}
\frac{z}{u-zv}\left(uR(u,v)-zvR(zv,v)\right)  -\frac{z}{u-zv}\left(uR(u,0)-zvR(zv,0)\right)\\+ \frac{z^2u}{v-zu}\left(R(u,v)-R(u,zu)\right).
\end{multline*}
\item {\bf Case 3.} The top side of the box was the last to move. Then the last inflating step was N. After this, the walk stepped E as far as the right side of the box. The generating function for such walks is
\[zvT(z,v).\]
\end{itemize}
Summing all these contributions, we obtain
\begin{multline*}
R(u,v) = 1 + \frac{za}{u-zv}\left(uR(u,0) - zvR(zv,0)\right) + \frac{z}{u-zv}\left(uR(u,v)-zvR(zv,v)\right) \\ -\frac{z}{u-zv}\left(uR(u,0)-zvR(zv,0)\right)+ \frac{z^2u}{v-zu}\left(R(u,v)-R(u,zu)\right) + zvT(z,v)
\end{multline*}
which is equivalent to~\eqref{eqn:R_eqn}.
\end{proof}

\begin{lem}\label{lem:it_ker_func_eqn}
The generating functions $T(u,v)$ and $R(u,v)$ satisfy the functional equation
\begin{equation}\label{eqn:R_eqn_elim_tvv_terms}
M(u,v)R(u,v) = A(u,v) + B(u,v)R(u,0)+C(u,v)T(z,v)
\end{equation}
where 
\begin{align*}
\lambda(v) &:= \lambda(z;v) = \frac{1+z^2-zv+z^3v-\sqrt{(1+z^2-zv+z^3v)^2-4z^2}}{2z}\\
A(u,v) & := A(z;u,v;a) = \frac{v(u-z^2u-zu\lambda(v)a+z^2v\lambda(v)a)}{(u-zv)(v-zu)(1-z\lambda(v)a)} \\
B(u,v) & := B(z;u,v;a) = -\frac{zu(u+v-zv-z^2v-va+z^2va)}{(u-zv)(v-zu)}\\
C(u,v) & := C(z;u,v) = -\frac{zv(zu-u\lambda(v)+zv\lambda(v))}{(\lambda(v)-z)(u-zv)}
\end{align*}
\end{lem}

\begin{proof}
Setting $v=0$ in \eqref{eqn:R_eqn} yields
\begin{equation*}
(1+z-za)R(u,0) = 1+zR(u,zu),
\end{equation*}
which can be used to eliminate $R(u,zu)$ from \eqref{eqn:R_eqn}. 
Meanwhile, $\lambda(v)$ satisfies the equation $L(\lambda(v),v)=0$ 
(as does another function of $v$, but $\lambda$ is the only root which is a power series 
in $z$ and thus the only one which will eventually give a well-defined solution), 
so substituting $u=\lambda(v)$ into \eqref{eqn:T_eqn} cancels the LHS. 
The resulting equation can be written as
\begin{equation}\label{eqn:Teqn_cancel_kernel}
-\frac{z^2v}{u-zv}R(zv,v)+\frac{z^2v (1-a)}{u-zv}R(zv,0) = \frac{zv}{(u-zv)(1-z\lambda(v)a)}-\frac{z^3v^2}{(u-zv)(\lambda(v)-z)}T(z,v).
\end{equation}
This can be used to eliminate $R(zv,v)$ and $R(zv,0)$ from \eqref{eqn:R_eqn}. 
Simple manipulation gives \eqref{eqn:R_eqn_elim_tvv_terms} in the stated form.
\end{proof}

\begin{prop}\label{prop:Ru0_soln}
The form of $R(u,0)$ is given by
\begin{equation}\label{eqn:Ru0_soln}
R(u,0) = \sum_{n=0}^{\infty} \mathcal{H}(u\Lambda^{2n})\prod_{k=0}^{n-1}\mathcal{I}(u\Lambda^{2k})
\end{equation}
when considered as a formal power series in $z,u,a$, where
\begin{align*}
\Lambda &:= \lambda(1) = \frac{1-z+z^2+z^3-\sqrt{1-2z-z^2-z^4+2z^5+z^6}}{2z}\\
\mathcal{H}(q) &:= \mathcal{H}(z;q;a) = \frac{1}{\mathcal{J}(a)}\cdot
\frac{(1-z^2)(1-\Lambda^2)(1-za(1+z)\lambda(q\Lambda)+za\lambda(q\Lambda)^2)}{(1-z\Lambda)(1-za\lambda(q\Lambda))(1-(1-z-z^2+\Lambda)\lambda(q\Lambda))}\\
\mathcal{I}(q) &:= \mathcal{I}(z;q;a) =
\frac{1}{\mathcal{J}(a)}\cdot\frac{(\Lambda-z)(1-z-z^2-a+z^2a+\Lambda)(z-(1-z\Lambda)\lambda(q\Lambda))}{z(1-z\Lambda)(1-(1-z-z^2+\Lambda)\lambda(q\Lambda))}\\
\mathcal{J}(a) &:= \mathcal{J}(z;a) = 1+\Lambda-z\Lambda-z^2\Lambda-\Lambda a+z^2\Lambda a
\end{align*}
\end{prop}

\begin{proof}
Equation \eqref{eqn:R_eqn_elim_tvv_terms} is susceptible to the iterative kernel method. 
The kernel $M$ is cancelled at $(u,v)=(u,u\Lambda)$ (and since $M(u,v)=M(v,u)$, 
it also disappears at $(u,v)=(v\Lambda,v)$), and so by substituting the pairs
\[(u,v) = (u,u\Lambda) \quad \text{and} \quad (u,v) = (u\Lambda^2,u\Lambda),\]
we obtain two equations in $R(u,0), R(u\Lambda^2,0)$ and $T(z,u\Lambda)$. 
Eliminating $T(z,u\Lambda)$ gives
\begin{equation}\label{eqn:itker_cancelT}
R(u,0) = \mathcal{H}(u) + \mathcal{I}(u)R(u\Lambda^2,0)
\end{equation}
where
\begin{align*}
\mathcal{H}(u) &= -\frac{A(u,u\Lambda)}{B(u,u\Lambda)} + \frac{C(u,u\Lambda)A(u\Lambda^2,u\Lambda)}{B(u,u\Lambda)C(u\Lambda^2,u\Lambda)}\\
\mathcal{I}(u) &= \frac{C(u,u\Lambda)B(u\Lambda^2,u\Lambda)}{B(u,u\Lambda)C(u\Lambda^2,u\Lambda)}
\end{align*}
After simplification $\mathcal{H}$ and $\mathcal{I}$ take the form given in the Proposition.

Iterating \eqref{eqn:itker_cancelT} gives \eqref{eqn:Ru0_soln}, provided that everything converges 
as a formal power series. But now $\mathcal{H}(u)$ is a power series in $z$ with coefficients 
in $\mathbb{Z}[u,a]$ of the form $1+za+O(z^2)$, and likewise $\mathcal{I}(u)$ is a power series 
of the form $z^4(1-a-u+ua) + O(z^5)$. As such, the sum converges as a formal power series in $z$ 
with coefficients in $\mathbb{Z}[u,a]$.
\end{proof}

\begin{thm}\label{thm:square_genfunc_solution}
The generating functions $R(u,v)$ and $T(u,v)$ have the solutions
\begin{align*}
R(u,v) &= \frac{1}{M(u,v)}\left[A(u,v)+B(u,v)R(u,0)+C(u,v)T(z,v)\right]\\
T(u,v) &= \frac{1}{L(u,v)}\left[\frac{1}{1-zua}-\frac{1}{1-z\lambda(v)a}+z^2v\left(\frac{1}{\lambda(v)-z} - \frac{1}{u-z}\right)T(z,v)\right]
\end{align*}
where
\[T(z,v) = \frac{-1}{C(v\Lambda,v)}\left[A(v\Lambda,v)+B(v\Lambda,v)R(v\Lambda,0)\right]\]
and $R(u,0), R(v\Lambda,0)$ are given by \eqref{eqn:Ru0_soln}. The overall generating function for 
2-sided prudent walks above an impenetrable surface is given by
\begin{align*}
W(u,v) := W(z;u,v;a) &= \sum_{n,i,j,\nu}W_{n,i,j,\nu}z^n u^i v^j a^\nu\\
&= R(u,v) + T(u,v) - T(0,v)
\end{align*}
where $W_{n,i,j,\nu}$ is the number of $n$-step 2-sided prudent walks which end a distance 
$i$ from the NE corner of their box and a distance $j$ above the surface, with $\nu$ steps along the surface.
\end{thm}

\begin{proof}
Substituting $(u,v) = (v\Lambda,v)$ into \eqref{eqn:R_eqn_elim_tvv_terms} cancels the LHS, 
and rearranging gives the stated form of $T(z,v)$. Since $v\Lambda$ is a power series in $z$ 
with coefficients in $\mathbb{Z}[v]$, $R(v\Lambda,0)$ is well-defined as a formal power series. 
Rearranging \eqref{eqn:R_eqn_elim_tvv_terms} then gives $R(u,v)$.

Equation \eqref{eqn:Teqn_cancel_kernel} can be rewritten as 
\[zR(zv,v) -z(1-a)R(zv,0) = \frac{-1}{1-z\lambda(v)a} + \frac{z^2v}{\lambda(v)-z}T(z,v),\]
and this can be used to eliminate $R(zv,v)$ and $R(zv,0)$ from \eqref{eqn:T_eqn}. 
Rearranging the resulting equation gives the stated expression for $T(u,v)$.

The overall generating function $W(u,v)$ is then found by adding the generating functions
 for walks ending on the right and the top of the box, and we subtract $T(0,v)$ (or equivalently, $R(0,v)$) because walks 
ending at the NE corner have been counted twice.
\end{proof}

We present here the functional equations for general (3- or 4-sided) prudent walks on the square lattice, 
though we are unable to solve them. Define the generating functions
\begin{align*}
R^*(u,v,w) &:= R^*(z;u,v,w;a) = \sum_{n,i,j,k,\nu}R^*_{n,i,j,k,y}z^n u^i v^j w^k a^\nu\\
T^*(u,v,w) &:= T^*(z;u,v,w;a) = \sum_{n,i,j,k,\nu}T^*_{n,i,j,k,y}z^n u^i v^j w^k a^\nu
\end{align*}
where $R^*_{n,i,j,k,\nu}$ counts $n$-step walks ending on the right of their box and $T^*_{n,i,j,k,\nu}$ 
counts those ending at the top of their box. In both cases $i$ is the distance from the endpoint to 
the NE corner of the box, $j$ is the distance from the endpoint to the surface, $k$ is the distance 
from the endpoint to the W side of the box, and $\nu$ is the number of occupied edges along the surface.

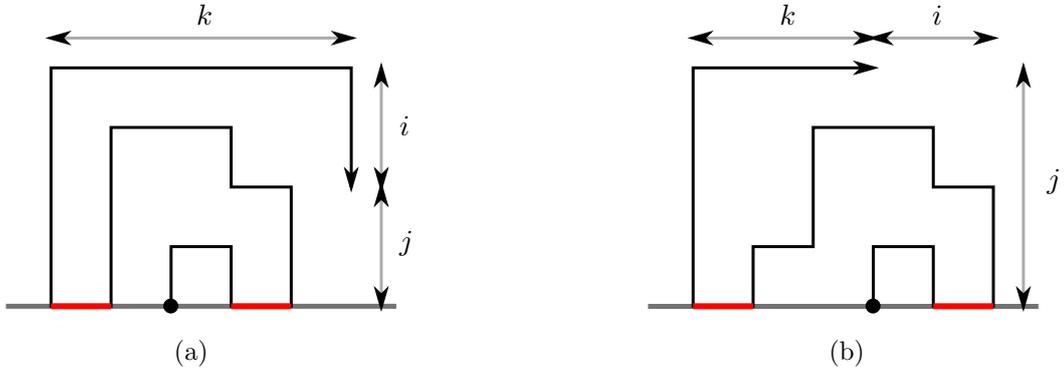
\begin{figure}
\centering
\begin{subfigure}[b]{0.4\textwidth}
\centering
\begin{picture}(150,90)
\put(0,0){\includegraphics[scale=1.4]{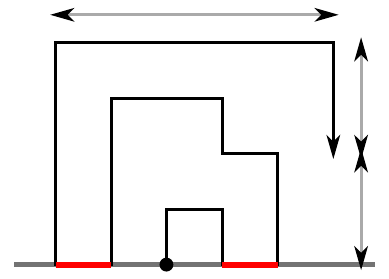}}
\put(153,70){$i$}
\put(153,27){$j$}
\put(77,112){$k$}
\end{picture}
\caption{}
\end{subfigure}
\hspace{2cm}
\begin{subfigure}[b]{0.4\textwidth}
\centering
\begin{picture}(160,90)
\put(0,0){\includegraphics[scale=1.4]{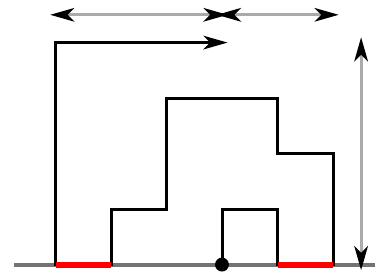}}
\put(155,50){$j$}
\put(55,112){$k$}
\put(112,112){$i$}
\end{picture}
\caption{}
\end{subfigure}
\caption{3-sided prudent walks above an impenetrable surface, (a) ending on the right of the box and 
(b) ending on the top, with the distances $i$, $j$ and $k$ indicated.}
\label{img:pru_sq_3sided_catalytic}
\end{figure}

\begin{lem} The generating functions $R^*(u,v,w)$ and $T^*(u,v,w)$ satisfy the equations
\begin{equation}\begin{split}
L^*(u,v,w)T^*(u,v,w) = 1 + zwR^*(zv,v,w) + zw(a-1)R^*(zv,0,w) + zuR^*(zv,v,u)\\ + zu(a-1)R^*(zv,0,u) - \frac{z^2vw}{u-zw}T^*(zw,v,w)-\frac{z^2uv}{w-zu}T^*(u,v,zu)
\end{split}\end{equation}
\begin{equation}\begin{split}
L^*(u,w,v)R^*(u,v,w) = 1 - \frac{z^2vw}{u-zv}R^*(zv,v,w) - \frac{z^2uw}{v-zu}R^*(u,zu,w)+\frac{zuw(a-1)}{u-zv}R^*(u,0,w)\\-\frac{z^2vw(a-1)}{u-zv}R^*(zv,0,w) + zvT^*(zw,v,w)
\end{split}\end{equation}
where
\[L^*(u,v,w):= L^*(z;u,v,w) = 1-\frac{zuvw(1-z^2)}{(u-zw)(w-zu)}.\]
\end{lem}
We omit the proof as the details follow in much the same way as Lemma~\ref{lem_square_original_func_eqns}. Here, inflating steps can be 
N, W or E; inflating steps to the N are attached to walks counted by $T^*$, and inflating steps 
to the W and E are attached to walks counted by $R^*$ (symmetry means $R^*$ also counts walks 
ending on the \emph{left} side of the bounding box).

Though the equations are somewhat more complicated, the problem we face here is similar to the 
difficulty of solving general prudent walks without a boundary~\cite{BousquetMelou2010Families} -- 
the presence of three catalytic variables seems to render the kernel method unable to produce 
a closed form expression.

\section{Singularities and asymptotic behaviour}

While we are quite certain of the location of the dominant singularity of $W(z;1,1;a)$ for all $a\geq 0$, 
the complicated nature of the generating functions has prevented us from obtaining a rigorous proof. 
The results of this section are thus stated as conjectures. They are followed at the end of the section by an outline of our reasoning.

\begin{conj}\label{conj:p:square_domsing_fe}
For a given $a > 0$, the dominant singularity of $W(z;1,1;a)$ is located at
\[\tSt(a) = \begin{cases}\tStone  & a\leq 2 \\ \tSttwo(a) & a>2,\end{cases}\]
where $\tStone\approx 0.403032$ is a root of $1-2z -2z^2+2z^3=0$, and $\tSttwo(a)$ is a root of

\begin{equation}\label{eqn:pdw_fe}
1-a-a(1-a)z +az^2+a(1-a)z^3=0.
\end{equation}
We thus observe a first-order adsorption phase transition occuring at $a=2$, with a crossover exponent of~$\phi=1$. 

\end{conj}
See Figure~\ref{fig:p:square_fe_density} for plots of the dominant singularity and surface density, as calculated using~\eqref{eqn:density} and~\eqref{eqn:fe_domsing}. It is of interest to note that the equation which determines the free energy in the adsorbed state, namely~\eqref{eqn:pdw_fe}, is the same for prudent tails and partially directed walks~\cite{Whittington1999Directedwalk}.

Series analysis of $W(z;1,1;a)$ for a variety of $a$ values confirms the validity of this conjecture -- see Figure~\ref{fig:exact_vs_estimates} for a comparison of the free energy and corresponding numerical estimates.
We point out that the dominant singularity in the desorbed phase, $\tStone\approx 0.403032$, is the same as for 
2-sided prudent walks \emph{without} a boundary~\cite{BousquetMelou2010Families}.

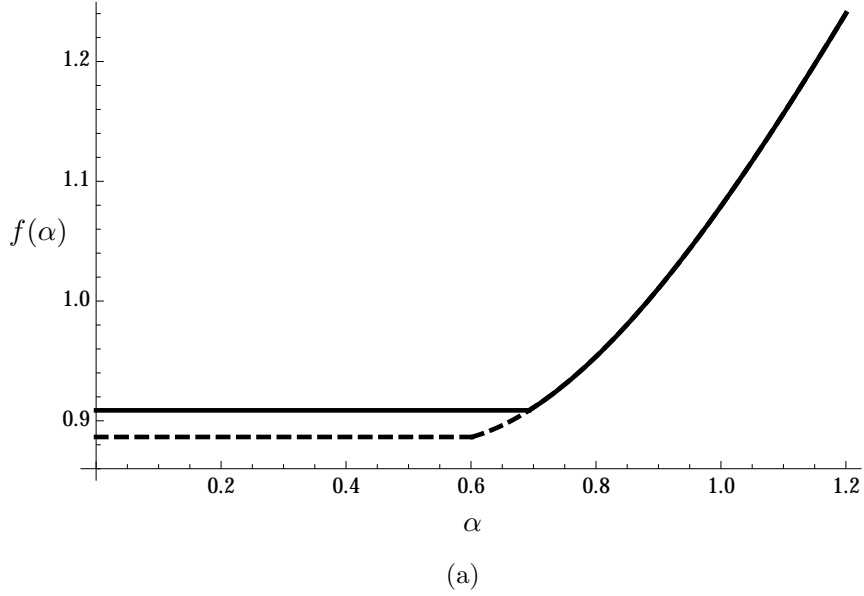
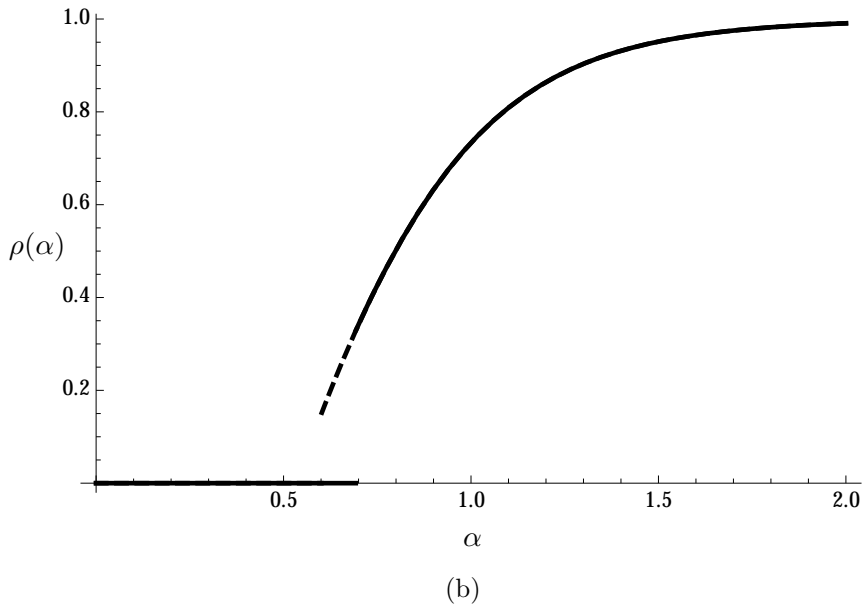
\begin{figure}[h!]
\centering
\begin{subfigure}[b]{0.8\textwidth}
\centering
  \begin{picture}(300,215)
    \put(0,15){\includegraphics[width=300pt]{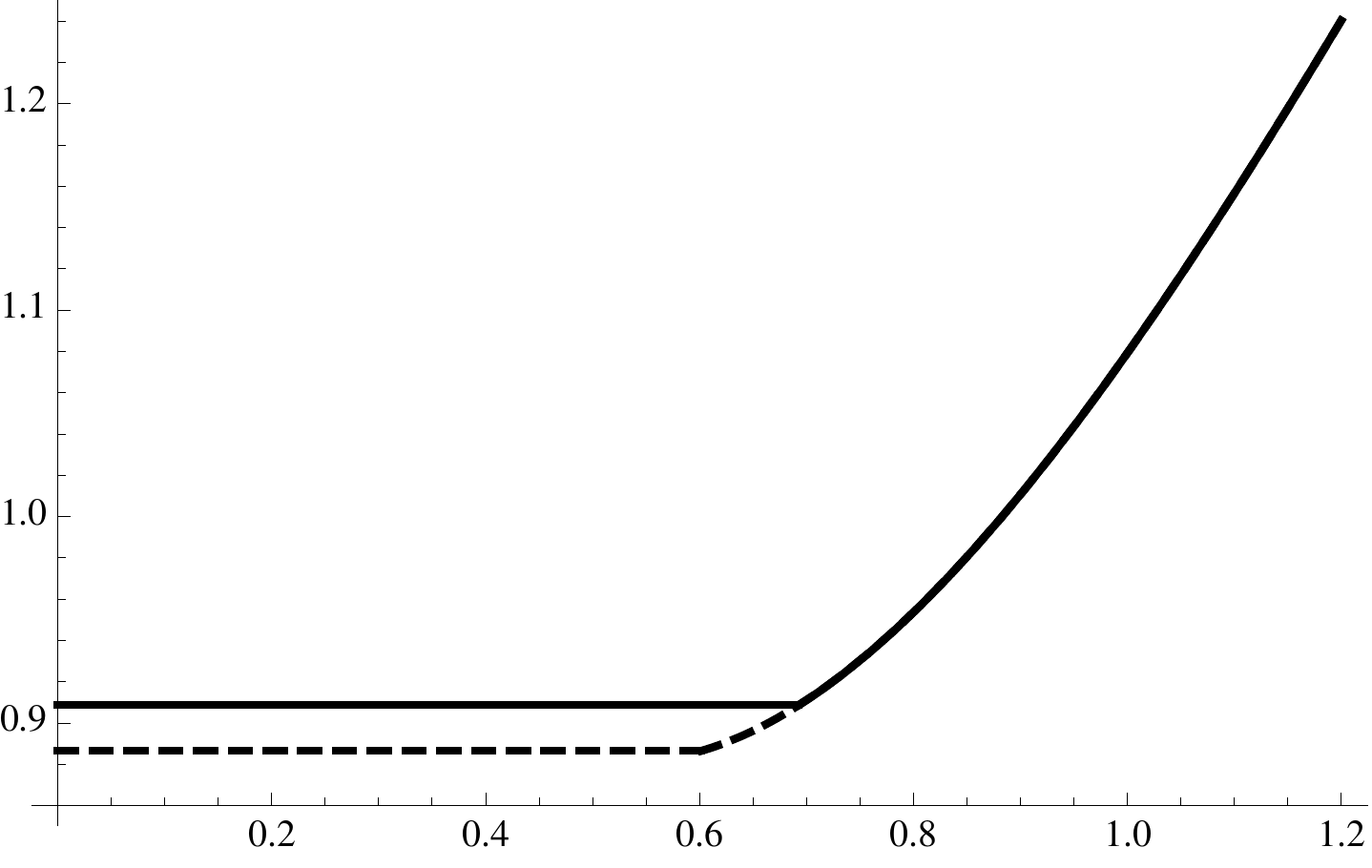}}
    \put(150,5){$\alpha$}
    \put(-20,115){$f(\alpha)$}
 \end{picture}
 \caption{}
 \end{subfigure}
\\
\vspace{1cm}
\begin{subfigure}[b]{0.8\textwidth}
\centering
  \begin{picture}(300,215)
    \put(0,15){\includegraphics[width=300pt]{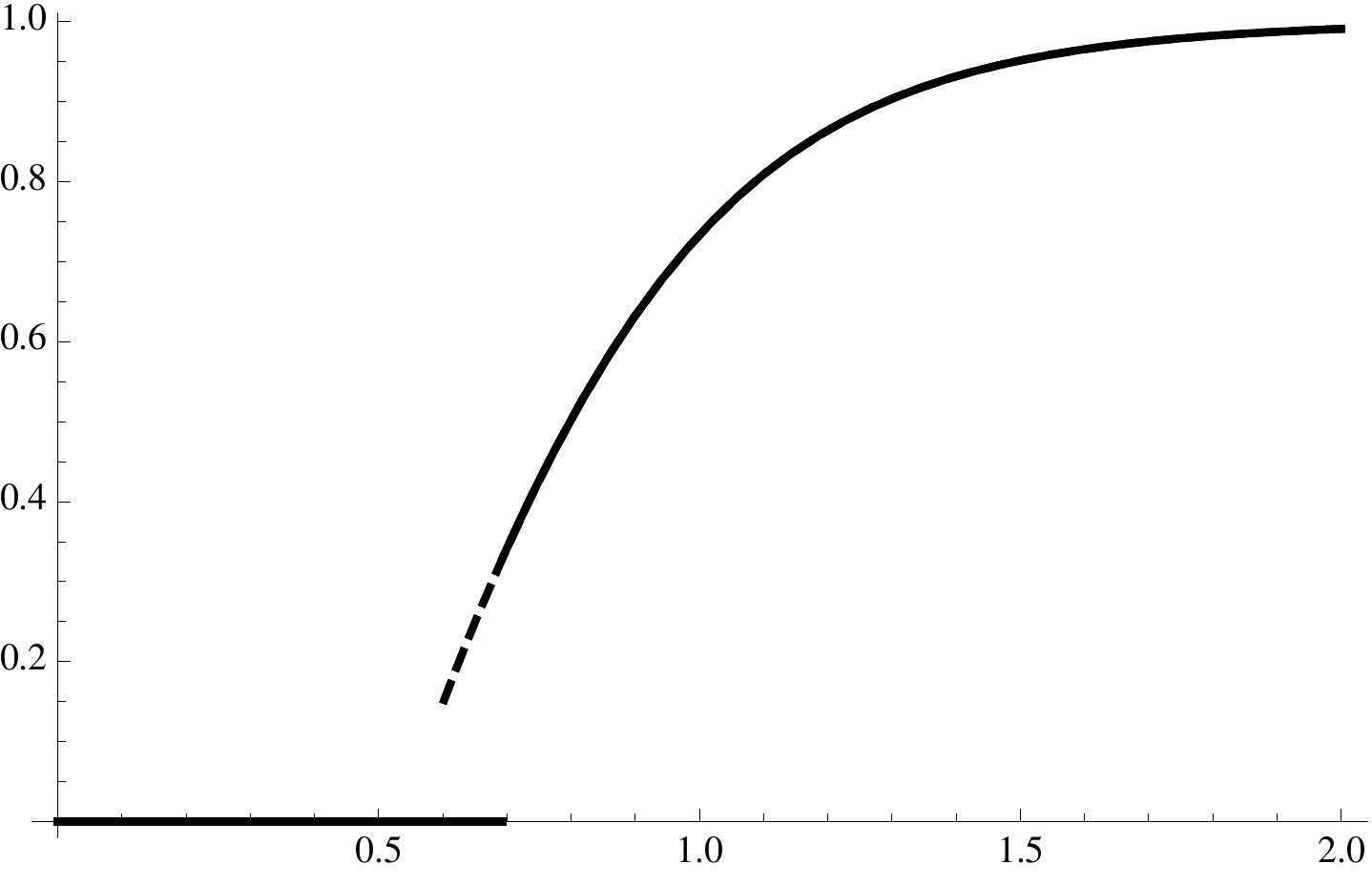}}
    \put(150,5){$\alpha$}
    \put(-20,115){$\rho(\alpha)$}
  \end{picture}
  \caption{}
  \end{subfigure}
\caption[Plots of the free energy and surface density of 2-sided prudent walks and loops.]
{Top: The (conjectured) free energy of 2-sided prudent walks (solid) and loops (dashed) as functions of the 
surface fugacity $\alpha = \log a$. 
Bottom: The (conjectured) density of edges in the surface for 2-sided prudent walks (solid) and loops (dashed). 
Note the discontinuity in the surface density for both models, corresponding to a first-order adsorption transition.}
\label{fig:p:square_fe_density}
\end{figure}

\begin{figure}
\centering
\begin{picture}(300,210)
    \put(0,10){\includegraphics[width=300pt]{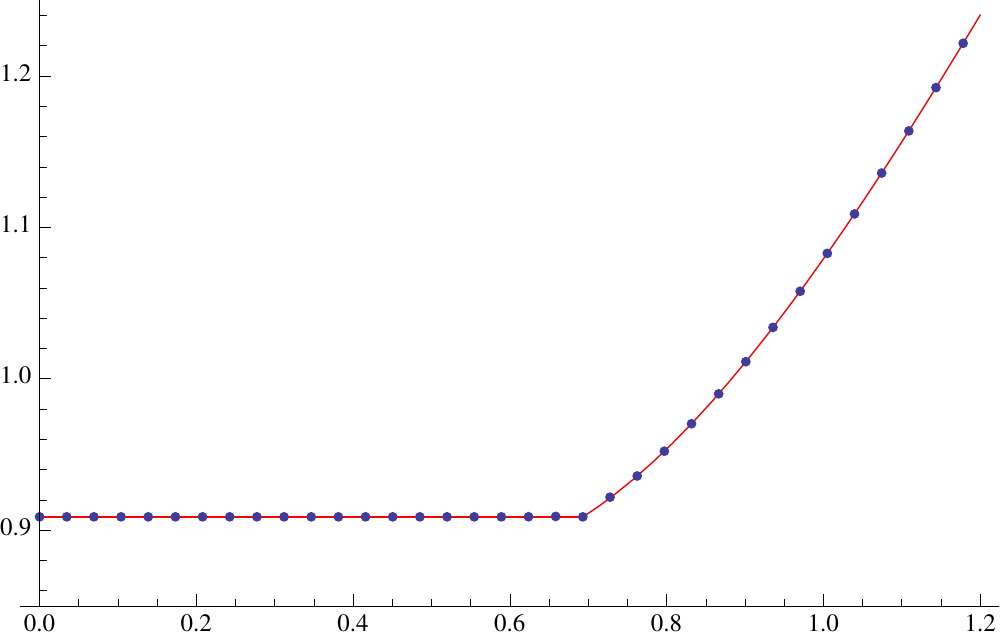}}
    \put(150,0){$\alpha$}
    \put(-20,110){$f(\alpha)$}
 \end{picture}
\caption{A plot of the conjectured free energy $f(\alpha)$ (the red line) along with numerical estimates (the blue dots) computed using series analysis of walks up to length 400.}
\label{fig:exact_vs_estimates}
\end{figure}

\begin{conj}\label{conj:p:square_loops_domsing}
For a given $a>0$, the dominant singularity of $W(z;1,0;a)$ is located at
\[\tSl(a) = \begin{cases} \tSlone  & a\leq \aSl\approx 1.82476 \\ \tSltwo(a) & a>\aSl,\end{cases}\]
where $\tSlone\approx 0.412095$ is a root of
\[1-3z-z^2+6z^3-7z^7-z^8+3z^9+z^{10}=0,\]
$\aSl$ is a root of
\[1-7a+45a^2-143a^3+277a^4-346a^5+285a^6-155a^7+54a^8-11a^9+a^{10}=0,\]
and $\tSltwo(a) = \tSttwo(a)$ as defined in Conjecture~\ref{conj:p:square_domsing_fe}. This implies that in the absorbed state, the free energy for prudent tails, prudent loops and partially directed walks is determined by the same equation, namely~\eqref{eqn:pdw_fe}. We thus observe a first-order adsorption transition occurring at 
$a=\aSl$ and a crossover exponent of $\phi=1$.

\end{conj}
(See Figure~\ref{fig:p:square_fe_density} for plots of the free energy and surface density.) 

The dominant singularity of $W(z;1,1;a)$ in the desorbed phase, namely $\tStone$, appears as a pole in the 
$T(z;z,1;a)$ term, which appears in both $R(z;1,1;a)$ and $T(z;1,1;a)$. More specifically, $\tStone$ is a root of 
$C(z;\Lambda,1)=0$. The dominant singularity in the adsorbed phase appears as a singularity of $R(z;1,0;a)$ and 
$R(z;\Lambda,0;a)$. In fact, $\tSttwo(a)$ is a root (in the variable $z$) of $\mathcal{J}(a)=0$ (as defined in 
Proposition~\ref{prop:Ru0_soln}), which in turn makes it a root of $B(z;1,\Lambda;a)=0$ (and in fact a root of 
$B(z;\Lambda^{n},\Lambda^{n+1};a)=0$ for $n\geq0$), and thus a pole of $\mathcal{H}(z;\Lambda^n;a)$ and 
$\mathcal{I}(z;\Lambda^n;a)$ for all $n\geq0$.

The trouble here lies not in \emph{locating} the singularity given in Conjecture~\ref{conj:p:square_domsing_fe}, 
but rather in \emph{proving} that $W(z;1,1;a)$ is analytic for $|z|<\tSt(a)$. The crux of the problem is that 
$\mathcal{J}(a)^{-n}$ appears in the summands of~\eqref{eqn:Ru0_soln}, and this means that $R(1,0)$, as it is defined there, 
is not \emph{absolutely-uniformly convergent} for $|z|<\tSt(a)$. 

This property can be satisfied for small 
$z$ (for example, taking $|z|<\tSt(c a)$ for a sufficiently large constant $c$), allowing for a resummation\footnote{We refer interested readers to~\cite[pp. 2274-5]{Beaton2011Enumeration}, where a similar problem was 
encountered in the enumeration of 3-sided prudent polygons by area. In that case, however, the generating function was a 
sum of \emph{rational} functions of a \emph{single variable}, and is thus much simpler than the problem in this case.}, 
but unfortunately it is unclear how to extend this result to $|z|<\tSt(a)$.

\section{Order of the phase transition}

Even though Conjectures~\ref{conj:p:square_domsing_fe} and~\ref{conj:p:square_loops_domsing} remain unproven at this point, 
the results of the preceding sections are satisfying: we have defined a new model of polymer adsorption, 
derived functional equations satisfied by corresponding generating functions, solved those functional equations, 
and determined the free energy of the model in the limit of infinitely long polymers.  This was the initial goal of 
this paper: to find a model more general than directed or partially directed walks, which is nevertheless still solvable. 
This models might then, in some sense, be `closer' to SAWs than previously-solved models. We now have the scheme
\begin{equation}
\label{eqn:model_scheme} 
\text{fully directed walks} \subset \text{partially directed walks} \subset \text{prudent walks} \subset \text{SAWs.} 
\end{equation}

We note here that all the existing models discussed in Section~\ref{ssec:exactly_solved} have second-order adsorption 
transitions~\cite{Privman1989Directed, vanRensburg1999Adsorbing, vanRensburg2010Adsorbing}. 
Moreover, the general SAW model in two dimensions is also conjectured~\cite{vanRensburg2004Multiple} 
to exhibit a second-order transition. Where it is known exactly, the order of these 
transitions is of second-order for both loops and tails, irrespective of the location of the endpoint of the walks.

Given~\eqref{eqn:model_scheme}, the nature of the adsorption phase transition for prudent walks (as described in 
Conjectures~\ref{conj:p:square_domsing_fe} and~\ref{conj:p:square_loops_domsing}) may at first seem puzzling. Prudent walks, which undergo a first-order transition, are both a subset and a superset of models which undergo second-order transitions.

A qualititative way of understanding this is to consider the relative location of the singularities for partially directed walks and the two prudent models. In the desorbed state, the dominant singularities are given by $z_1^{t}<z_1^{\ell}<z_1^{PDW}=\sqrt{2}-1$.  On the other hand, in the adsorbed state, the equation determining the free energy is given by the solution of the same equation for all three models, namely~\eqref{eqn:pdw_fe}.  Combining this with the fact that partially directed walks undergo a second order transition necessarily implies that the two prudent walk models must undergo a first order transition at their respective critical adsorption values.

The following (rough) conjecture is another approach to understanding this feature.

\begin{conj}\label{conj:transitions_height}
For a given (infinite) subclass $\mathcal{C}$ of self-avoiding walks in an upper half-plane with starting point on the surface, 
let $\omega\in\mathcal{C}$ and define $h(\omega)$ to be the maximum height above the surface reached by $\omega$. 
Then let $\langle h_n \rangle$ be the average of $h(\omega)$ over all walks $\omega\in\mathcal{C}$ of length $n$. 
If walks in $\mathcal{C}$, equipped with a surface interaction associated with steps along the surface, 
exhibit an adsorption phase transition, then the transition will be first-order if $\langle h_n \rangle = \Theta(n)$, 
and second-order if $\langle h_n \rangle = O(n^{3/4})$.
\end{conj}

The exponent $3/4$ is widely expected~\cite{Nienhuis1982Exact} to be the value $\nu$ characterising the size of 
self-avoiding walks in two dimensions. At present there are no known models with an exponent $\nu$ strictly between 
$3/4$ and 1, and so we do not attempt to cover this range in the conjecture.

The intuitive reasoning behind this conjecture is as follows. A first-order phase transition characterises a sudden, 
violent change -- a system on one side of the transition point is markedly different from one the other. 
Second-order transitions, while still characterising a change, are smoother, and less sudden.

If a given walk model has $\langle h_n \rangle = \Theta(n)$, then when desorbed, the bulk of an average walk must be 
far removed from the surface. But when adsorption takes place, a significant portion of the walk must now be very close 
to the surface. This dichotomy between the desorbed and adsorbed regimes suggests a sudden, noticeable change must take 
place at the critical point, implying a first-order transition.

On the other hand, when $\langle h_n \rangle = O(n^{3/4})$, the distance between the bulk of a walk and the surface 
is small, relative to the length of the walk. Thus the difference between the adsorbed and desorbed phases 
(at least, close to the critical point) is less dramatic, suggesting a second-order transition.

\subsection{Directed models}\label{sec:directed_models_heights}

In this section we briefly consider Conjecture~\ref{conj:transitions_height} in the context of existing, directed models. 
We will discuss four such models, all restricted to the upper half-plane, with an impenetrable surface in the horizontal line 
$y=0$.
\begin{itemize}
\item NE-directed walks, i.e. walks on the square lattice consisting only of north and east steps;
\item Dyck paths, i.e. walks on a $45^\circ$ rotated square lattice (or, just as easily, the triangular lattice) 
consisting of north-east and south-east steps;
\item Motzkin paths, i.e. walks on the triangular lattice consisting of north-east, east, and south-east steps; and
\item partially directed walks, i.e. walks on the square lattice consisting of north, east and south steps.
\end{itemize}
In the literature, the terms \emph{Dyck path} and \emph{Motzkin path} are usually used to refer to walks which start 
\emph{and end} on the surface. To avoid confusion we will instead use the terms \emph{loops} and \emph{tails} to refer to 
walks which end on the surface or at any height, respectively.

In the following table we summarise the results for the asymptotics of the average height $\langle h_n \rangle$ of walks of 
length $n$, as well as the order of the adsorption phase transition. (Note that for models which allow steps parallel to 
the surface, one can use edge- or vertex-weights on the surface, but this does not affect the order of the phase transition.) 
The calculations for NE-directed walks are trivial, and the average heights given there are exact for all $n$, not just in the 
large $n$ limit. The numerical estimates are ours, and are derived from analysis of long (4000 or 5000 terms) series.

\begin{center}
\begin{tabular}{| p{0.19\textwidth} || p{0.36\textwidth} | p{0.36\textwidth} |}
\hline
{\bf Model} & {\bf Loops} & {\bf Tails} \\
\hline \hline
NE-directed & 0 \newline no phase transition & $n/2$ \newline first-order \\
\hline
Dyck~\cite{deBruijn1972Average, Kemp1981Average, vanRensburg1999Adsorbing} & $\sqrt{\frac{\pi}{2}}n^{1/2} - \frac{3}{2} + O\left(n^{-1/2}\right)$ \newline second-order & $\left(\sqrt{2\pi}\log 2\right)n^{1/2} -\frac{3}{2} + O(n^{-1/2})$ \newline second-order \\
\hline
Motzkin~\cite{Prodinger1980Average, vanRensburg2010Adsorbing} & $\sqrt{\frac{\pi}{3}}n^{1/2} - \frac{3}{2} + O\left(n^{-1/2}\right)$ \newline second-order & $1.418632\,n^{1/2} - 1.5000 + O(n^{-1/2})$ \newline second-order \\
\hline
partially directed \newline \cite{Privman1989Directed} & $1.376996\,n^{1/2} - 2.91422 + O(n^{-1/2})$ \newline second-order & $1.908922\,n^{1/2} - 2.91422 + O(n^{-1/2})$ \newline second-order \\
\hline
\end{tabular}
\end{center}

We note that for Motzkin paths and partially directed walks, the ratio of the leading amplitudes of tails and loops seems to be 
the same as that of Dyck paths, namely $2\log 2 \approx 1.38229$. This leads us to conjecture that the average height 
$\langle h_n\rangle$ of Motzkin tails of length $n$ is asymptotically
\[\langle h_n \rangle = \left(2\sqrt{\frac{\pi}{3}}\log 2\right)n^{1/2} - \frac{3}{2} + O(n^{-1/2}).\]
If true, this can very likely be obtained in the same manner as the exact results given above.

More importantly to us, we see that all of these models (except the trivial case of NE-directed ``loops'', which display 
no transition at all) obey Conjecture~\ref{conj:transitions_height}. The single model with average height $\Theta(n)$, 
NE-directed tails, undergoes a first-order adsorption transition, while all the others, having average height 
$\Theta(n^{1/2})$, display second-order transitions.

\subsection{Prudent models}

We now return to prudent walks, and investigate the average height of walks of length $n$ and how this relates to the order of 
the adsorption phase transition. We begin with a brief outline of how the asymptotics of the average height 
$\langle h_n \rangle$ can be calculated.

Let $c_n(\lambda)$ be the number of walks in a class $\mathcal{C}$ which reach a maximum height $\lambda$ above the surface, 
and define the bivariate generating function 
\[C(z;v) = \sum_{n,\lambda} c_n(\lambda) z^n v^\lambda.\]
Say $C(z;1)$ has radius of convergence $\rho$. Then
\[\langle h_n\rangle = \frac{\sum_{\lambda} \lambda c_n(\lambda)}{\sum_{\lambda} c_n(\lambda)} = \frac{\frac{\partial}{\partial v}[z^n]C(z;v)|_{v=1}}{[z^n]C(z;1)} = \frac{[z^n]\frac{\partial}{\partial v}C(z;v)|_{v=1}}{[z^n]C(z;1)},\]
where the last equality follows from the uniform convergence of $C(z;v)$ on (at least) $(z,v) \in [0,\rho) \times [0,1]$. 
For all the models discussed in this article, we expect to find
\[[z^n]C(z;1) \sim An^{\gamma}\rho^{-n} \qquad \text{and}\qquad [z^n]\left.\frac{\partial}{\partial v}C(z;v)\right|_{v=1} \sim Bn^{\gamma'}\rho^{-n}\]
for constants $A,B,\gamma,\gamma'$, implying
\[\langle h_n \rangle = \Theta(n^{\gamma'-\gamma}).\]
As per Conjectures~\ref{conj:p:square_domsing_fe} and~\ref{conj:p:square_loops_domsing}, we expect that both tails and loops 
display first-order transitions. For tails, it is easier to consider the average height of the endpoints of walks, 
$\langle e_n \rangle$, rather than the maximum height. Note that since the maximum height of any walk is at least as great its endpoint, we have $\langle h_n \rangle \geq \langle e_n \rangle$. 

The variable $v$ in $W(z;1,v;1)$ tracks the endpoint height of walks. It thus appears that we must investigate the asymptotics of 
\[[z^n]\left.\frac{\partial}{\partial v}W(z;1,v;1)\right|_{v=1}.\]
However, we are fortunate in that $W$ satisfies a set of properties which allow us to sidestep this analysis. Theorem~IX.9 of~\cite{Flajolet2009Analytic} states that if a generating function $F(z;u)$ satisfies a certain \emph{meromorphic schema}, then the random variable $X_n$ with probability generating function
\[p_n(u) = \frac{[z^n]F(z,u)}{[z^n]F(z,1)}\]
has mean and standard deviation which are asymptotically $\Theta(n)$. The random variable of interest here is $e_n$, and it has probability generating function
\[p_n(v) = \frac{[z^n]W(z;1,v;1)}{[z^n]W(z;1,1;1)}.\]
We thus need to show that $W(z;1,v;1)$ satisfies the three conditions set out in Theorem~IX.9 of~\cite{Flajolet2009Analytic}. For completeness we state the details of that theorem in the Appendix.

 Condition (\emph{i}) amounts to showing that $W(z;1,v;1)$ has an isolated simple pole at $(z,v) = (z_1^t,1)$. 

One can factor out the singular term $C(z;v\Lambda,v)^{-1}$ from the
generating function and show that what remains is analytic on, for example, 
\[(z,v)\in\mathcal D^t = \{|z|\leq 0.405\}\times\{|v-1|<0.01\}.\] 
(Recall that the dominant singularity at $v=1$ is a simple pole at $\tStone\approx0.403$.) The most complicated part of this calculation involves showing that the infinite sum~\eqref{eqn:Ru0_soln} is absolutely-uniformly convergent on $\mathcal D^t$ for $u=1$ and for $u=v\Lambda$, but even this only requires one to derive simple bounds for the individual factors of $\mathcal H$ and $\mathcal I$.

Condition (\emph{ii}) is a straightforward derivative, and Condition (\emph{iii}) can be checked via implicit differentiation of $C(z;v\Lambda,v)$. It follows that $\langle e_n \rangle = \Theta(n)$, and thus the same is true of 
$\langle h_n \rangle$. 

For prudent loops, the height is tracked by the variable $u$ in $R(z;u,0;1)$. The dominant singularity at $u=1$ is 
again a simple pole, and so we are still able to use Theorem~IX.9 of~\cite{Flajolet2009Analytic}. The singular term 
$(1-(1-z-z^2+\Lambda)\lambda(u\Lambda))^{-1}$ can be factored out, and what remains can be shown to be analytic on, 
for example, 
\[(z,u) \in \mathcal D^\ell = \{|z|\leq 0.413\}\times\{|u-1|<0.001\}.\]
(The dominant pole is at $\tSlone\approx0.412$.) Conditions (\emph{ii}) and (\emph{iii}) are again straightforward. 
We thus have $\langle h_n \rangle = \Theta(n)$.

We summarise these results in the following lemma, which is given without proof.

\begin{lem}\label{lem:average_height}
The average endpoint height $\langle e_n\rangle$ and total height $\langle h_n \rangle$ of $n$-step 2-sided prudent tails in the upper half-plane are both asymptotically linear in $n$. The average total height $\langle h_n\rangle$ of $n$-step 2-sided prudent loops in the upper half-plane is also asymptotically linear in $n$.
\end{lem}

One way of understanding the above lemma is to note that prudent walks, although seemingly undirected, feel the effect of a ``pseudo-force''. The step set of prudent walks is unrestricted in the sense that the walk is able to step in all four lattice directions, while maintaining the prudent condition. However, a 2-sided prudent walk is unable to take the pairs of steps south-west or west-south. This has the effect of biasing the step set in a direction at an angle to the surface. This bias can be interpreted as a ``force'' acting on the end of the polymer, pulling it in the direction of the pseudo-force.

The effect of a pulling force on a polymer has been considered in the case of self-avoiding walks~\cite{Guttmann2014Pulling} and partially directed walks~\cite{Orlandini2010Adsorbing, Osborn_2010}. When the component of the force perpendicular to the surface is sufficiently large, desorption by the \emph{force} occurs and the transition is first-order for both models. Furthermore, in the case of partially directed walks, it is known that there exists a critical value of the angle, below which adsorption is \emph{enhanced}, and desorption cannot be induced by the force. In this regime, \emph{thermal} desorption can still occur, and that transition is second-order.

This ``pseudo-force'' may be one mechanism which leads to the observed first-order transitions for 2-sided prudent walks.

\section{Conclusion}\label{sec:conclusion}
In this paper we have studied two related models of a polymer interacting with an impenetrable surface. The lattice walks considered, namely 2-sided prudent tails and loops, are able to step in all four directions on the lattice, unlike many previously-studied directed models like Dyck paths and partially directed walks. In order to analyze these walks we have used the generating function of 2-sided prudent walks in the upper half-plane. The solution to this generating function is obtained by applying the iterated kernel method to a certain functional equation. By studying the singularity structure of this generating function, we obtain thermodynamic quantities for prudent tails and loops.

 In both cases we locate the critical value of the interaction parameter, and observe a first-order adsorption transition. We relate the order of the phase transition in this model with those of other models like adsorbing Dyck paths and SAWs, and argue that the difference arises from a ``pseudo-force'' imposed by the restricted step set of prudent walks.

\appendix
\section*{Appendix}
\renewcommand{\thesubsection}{\Alph{subsection}}

\begin{thm*}[Theorem IX.9 of~\cite{Flajolet2009Analytic}]
Let $F(z,u)$ be a function that is bivariate analytic at $(z,u) = (0,0)$ and has non-negative coefficients. Assume that $F(z,1)$ is meromorphic in $z\leq r$ with only a simple pole at $z=\rho$ for some positive $\rho<r$. Assume also the following conditions.
\begin{itemize}
\item[(i)] \emph{Meromorphic perturbation:} there exists $\epsilon >0$ and $r>\rho$ such that in the domain $\mathcal D = \{|z|\leq r\}\times\{|u-1|<\epsilon\}$, the function $F(z,u)$ admits the representation
\[F(z,u) = \frac{B(z,u)}{C(z,u)},\]
where $B(z,u)$, $C(z,u)$ are analytic for $(z,u)\in\mathcal D$ with $B(\rho,1)\neq 0$. (Thus $\rho$ is a simple zero of $C(z,1)$.)
\item[(ii)] \emph{Non-degeneracy:} one has $\partial_z C(\rho,1)\cdot \partial_u C(\rho,1)\neq 0$, ensuring the existence of a non-constant $\rho(u)$ analytic at $u=1$, such that $C(\rho(u),u)=0$ and $\rho(1)=\rho$.
\item[(iii)] \emph{Variability:} one has
\[\mathfrak v\left(\frac{\rho(1)}{\rho(u)}\right)\neq 0,\]
where $\mathfrak v(g(u)) = g''(1) + g'(1)-g'(1)^2$.
\end{itemize}
Then the random variable $X_n$ with probability generating function
\[p_n(u) = \frac{[z^n]F(z,u)}{[z^n]F(z,1)}\]
after standardization, converges in distribution to a Gaussian variable, with a speed of convergence that is $O(n^{-1/2})$. The mean and standard deviation of $X_n$ are asymptotically linear in $n$.
\end{thm*}

\bibliography{pru_surface}

\begin{bibdiv}
\begin{biblist}

\bib{Beaton2011Enumeration}{article}{
      author={Beaton, N.~R.},
      author={Flajolet, P.},
      author={Guttmann, A.~J.},
       title={The enumeration of prudent polygons by area and its unusual
  asymptotics},
        date={2011},
        ISSN={00973165},
     journal={Journal of Combinatorial Theory, Series A},
      volume={118},
       pages={2261\ndash 2290},
         url={http://dx.doi.org/10.1016/j.jcta.2011.05.004},
}

\bib{BousquetMelou2010Families}{article}{
      author={Bousquet-M\'{e}lou, M.},
       title={Families of prudent self-avoiding walks},
        date={2010},
        ISSN={00973165},
     journal={Journal of Combinatorial Theory, Series A},
      volume={117},
       pages={313\ndash 344},
         url={http://dx.doi.org/10.1016/j.jcta.2009.10.001},
}

\bib{Brak1992Collapse}{article}{
      author={Brak, R.},
      author={Guttmann, A.~J.},
      author={Whittington, S.~G.},
       title={A collapse transition in a directed walk model},
        date={1992},
        ISSN={0305-4470},
     journal={Journal of Physics A: Mathematical and General},
      volume={25},
       pages={2437+},
         url={http://dx.doi.org/10.1088/0305-4470/25/9/017},
}

\bib{Clisby_2013}{article}{
      author={Clisby, N.},
       title={Calculation of the connective constant for self-avoiding walks
  via the pivot algorithm},
        date={2013},
     journal={Journal of Physics A: Mathematical and Theoretical},
      volume={46},
      number={24},
       pages={245001},
         url={http://dx.doi.org/10.1088/1751-8113/46/24/245001},
}

\bib{Clisby2012New}{article}{
      author={Clisby, N.},
      author={Jensen, I.},
       title={A new transfer-matrix algorithm for exact enumerations:
  self-avoiding polygons on the square lattice},
        date={2012},
        ISSN={1751-8113},
     journal={Journal of Physics A: Mathematical and Theoretical},
      volume={45},
      number={11},
       pages={115202+},
         url={http://dx.doi.org/10.1088/1751-8113/45/11/115202},
}

\bib{deBruijn1972Average}{incollection}{
      author={de~Bruijn, N.~G.},
      author={Knuth, D.~E.},
      author={Rice, S.~O.},
       title={The average height of planted plane trees},
        date={1972},
   booktitle={Graph {T}heory and {C}omputing},
      editor={Read, R.~C.},
   publisher={Academic Press},
       pages={15\ndash 22},
}

\bib{Duchi2005Some}{inproceedings}{
      author={Duchi, E.},
       title={On some classes of prudent walks},
        date={2005},
   booktitle={{FPSAC}'05},
}

\bib{DuminilCopin2012Connective}{article}{
      author={Duminil-Copin, H.},
      author={Smirnov, S.},
       title={The connective constant of the honeycomb lattice equals
  $\sqrt{2+\sqrt{2}}$},
        date={2012},
        ISSN={0003-486X},
     journal={Annals of Mathematics},
      volume={175},
       pages={1653\ndash 1665},
         url={http://dx.doi.org/10.4007/annals.2012.175.3.14},
}

\bib{Flajolet2009Analytic}{book}{
      author={Flajolet, P.},
      author={Sedgewick, R.},
       title={Analytic {C}ombinatorics},
   publisher={Cambridge University Press},
        date={2009},
        ISBN={9780521898065},
         url={http://www.worldcat.org/isbn/9780521898065},
}

\bib{Guttmann2014Pulling}{article}{
      author={Guttmann, A.~J.},
      author={Jensen, I.},
      author={Whittington, S.~G.},
       title={Pulling adsorbed self-avoiding walks from a surface},
        date={2014},
     journal={Journal of Physics A: Mathematical and Theoretical},
      volume={47},
       pages={015004},
}

\bib{Hammersley1982Selfavoiding}{article}{
      author={Hammersley, J.~M.},
      author={Torrie, G.~M.},
      author={Whittington, S.~G.},
       title={Self-avoiding walks interacting with a surface},
        date={1982},
        ISSN={0305-4470},
     journal={Journal of Physics A: Mathematical and General},
      volume={15},
       pages={539+},
         url={http://dx.doi.org/10.1088/0305-4470/15/2/023},
}

\bib{vanRensburg1999Adsorbing}{article}{
      author={Janse~van Rensburg, E.~J.},
       title={Adsorbing staircase walks and staircase polygons},
        date={1999},
        ISSN={0218-0006},
     journal={Annals of Combinatorics},
      volume={3},
       pages={451\ndash 473},
         url={http://dx.doi.org/10.1007/BF01608797},
}

\bib{vanRensburg2010Adsorbing}{article}{
      author={Janse~van Rensburg, E.~J.},
       title={Adsorbing {M}otzkin paths},
        date={2010},
        ISSN={1751-8113},
     journal={Journal of Physics A: Mathematical and Theoretical},
      volume={43},
       pages={485006+},
         url={http://dx.doi.org/10.1088/1751-8113/43/48/485006},
}

\bib{vanRensburg2004Multiple}{article}{
      author={Janse~van Rensburg, E.~J.},
      author={Rechnitzer, A.},
       title={Multiple {M}arkov chain {M}onte {C}arlo study of adsorbing
  self-avoiding walks in two and in three dimensions},
        date={2004},
        ISSN={0305-4470},
     journal={Journal of Physics A: Mathematical and General},
      volume={37},
       pages={6875+},
         url={http://dx.doi.org/10.1088/0305-4470/37/27/002},
}

\bib{Jensen2004Selfavoiding}{article}{
      author={Jensen, I.},
       title={Self-avoiding walks and polygons on the triangular lattice},
        date={2004},
     journal={Journal of Statistical Mechanics: Theory and Experiment},
      volume={2004},
       pages={P10008+},
         url={http://dx.doi.org/10.1088/1742-5468/2004/10/P10008},
}

\bib{Jensen1999Selfavoiding}{article}{
      author={Jensen, I.},
      author={Guttmann, A.~J.},
       title={Self-avoiding polygons on the square lattice},
        date={1999},
     journal={Journal of Physics A: Mathematical and General},
      volume={32},
       pages={4867\ndash 4876},
         url={http://dx.doi.org/10.1088/0305-4470/32/26/305},
}

\bib{Kemp1981Average}{article}{
      author={Kemp, R.},
       title={On the average depth of a prefix of the {D}ycklanguage ${D}_1$},
        date={1981},
     journal={Discrete Mathematics},
      volume={36},
       pages={155\ndash 170},
}

\bib{Madras2003Localization}{article}{
      author={Madras, N.},
      author={Whittington, S.~G.},
       title={Localization of a random copolymer at an interface},
        date={2003},
        ISSN={0305-4470},
     journal={Journal of Physics A: Mathematical and General},
      volume={36},
       pages={923+},
         url={http://dx.doi.org/10.1088/0305-4470/36/4/305},
}

\bib{Nienhuis1982Exact}{article}{
      author={Nienhuis, B.},
       title={Exact critical point and critical exponents of {O}(n) models in
  two dimensions},
        date={1982},
     journal={Physical Review Letters},
      volume={49},
      number={15},
       pages={1062\ndash 1065},
         url={http://dx.doi.org/10.1103/PhysRevLett.49.1062},
}

\bib{Orlandini2010Adsorbing}{article}{
      author={Orlandini, E.},
      author={Whittington, S.~G.},
       title={Adsorbing polymers subject to an elongational force: the effect
  of pulling direction},
        date={2010},
        ISSN={1751-8113},
     journal={Journal of Physics A: Mathematical and Theoretical},
      volume={43},
      number={48},
       pages={485005+},
         url={http://dx.doi.org/10.1088/1751-8113/43/48/485005},
}

\bib{Orr1947Statistical}{article}{
      author={Orr, W. J.~C.},
       title={Statistical treatment of polymer solutions at infinite dilution},
        date={1947},
     journal={Transactions of the Faraday Society},
      volume={43},
       pages={12\ndash 27},
         url={http://dx.doi.org/10.1039/tf9474300012},
}

\bib{Osborn_2010}{article}{
      author={Osborn, J.},
      author={Prellberg, T.},
       title={Forcing adsorption of a tethered polymer by pulling},
        date={2010},
     journal={J. Stat. Mech.},
      volume={2010},
      number={09},
       pages={P09018},
         url={http://dx.doi.org/10.1088/1742-5468/2010/09/p09018},
}

\bib{Owczarek2007Exact}{article}{
      author={Owczarek, A.~L.},
      author={Prellberg, T.},
       title={Exact solution of semi-flexible and super-flexible interacting
  partially directed walks},
        date={2007},
        ISSN={1742-5468},
     journal={Journal of Statistical Mechanics: Theory and Experiment},
      volume={2007},
       pages={P11010+},
         url={http://dx.doi.org/10.1088/1742-5468/2007/11/P11010},
}

\bib{Prea1997Exterior}{unpublished}{
      author={Pr\'{e}a, P.},
       title={Exterior self-avoiding walks on the square lattice},
        date={1997},
        note={Unpublished},
}

\bib{Privman1989Directed}{book}{
      author={Privman, V.},
      author={\v{S}vrakic, N.~M.},
       title={Directed models of polymers, interfaces, and clusters: scaling
  and finite-size properties},
      series={Springer Lecture Notes in Physics},
   publisher={Springer-Verlag},
        date={1989},
        ISBN={0387514295},
         url={http://www.worldcat.org/isbn/0387514295},
}

\bib{Prodinger1980Average}{article}{
      author={Prodinger, H.},
       title={The average height of a stack where three operations are allowed
  and some related problems},
        date={1980},
     journal={Journal of Combinatorics, Information and System Sciences},
      volume={5},
       pages={287\ndash 304},
}

\bib{Whittington1975Selfavoiding}{article}{
      author={Whittington, S.~G.},
       title={Self-avoiding walks terminally attached to an interface},
        date={1975},
     journal={The Journal of Chemical Physics},
      volume={63},
       pages={779\ndash 785},
         url={http://dx.doi.org/10.1063/1.431357},
}

\bib{Whittington1999Directedwalk}{article}{
      author={Whittington, S.~G.},
       title={A directed-walk model of copolymer adsorption},
        date={1999},
        ISSN={0305-4470},
     journal={Journal of Physics A: Mathematical and General},
      volume={31},
      number={44},
       pages={8797+},
         url={http://dx.doi.org/10.1088/0305-4470/31/44/008},
}

\end{biblist}
\end{bibdiv}

\end{document}